\documentclass[runningheads]{llncs}
\usepackage[T1]{fontenc}
\usepackage{graphicx}
\usepackage{orcidlink}
\usepackage{todonotes}
\usepackage{amsmath}
\usepackage{amssymb}
\usepackage{martinmakros}
\usepackage{martin_eng}

\newcommand{\nexta}[1]{\langle #1 \rangle}

\newcommand{\Uu}{\mathcal{U}}

\newcommand{\moveto}[1]{{\stackrel{#1}{\longrightarrow}}_T}
\newcommand{\conf}{\mathcal{C}_T}
\newcommand{\View}{\downarrow\!\!}             %

\newcommand{\True}{\mathrm{tt}}
\newcommand{\False}{\mathrm{ff}}

\newcommand{\LTL}{\ensuremath{\mathrm{LTL}}\xspace}

\newcommand{\nega}{\ensuremath{{\neg}}}

\newcommand{\langP}[1]{{\mathcal{L}(#1)}}

\newcommand{\TR}{\mathbb{TR}}
\newcommand{\letters}[1]{alph(#1)}
\newcommand{\con}[1]{c_{#1}}
\newcommand{\prf}{\mathrm{prf}}
\newcommand{\lin}{\mathrm{lin}}
\newcommand{\linP}[1]{\mathrm{lin}(#1)}
\newcommand{\proj}{{\upharpoonright}}

\newcommand{\alphabet}{\ensuremath{\Sigma}}         %
\newcommand{\ir}{\ensuremath{I}}                    %
\newcommand{\w}{\ensuremath{w}}                     %

\newcommand{\E}{\ensuremath{E}}

\newcommand{\T}{\ensuremath{T}}         %

\newcommand{\cL}{\mathcal{L}}

\newcommand{\cPhat}{{\hat{\mathcal{P}}}}

\newcommand{\LTLthree}{\ensuremath{{\textrm{LTL}_3}}\xspace}
\newcommand{\LTLthreeTraces}{\ensuremath{{\textrm{LTrL}_3}}\xspace}

\newcommand{\B}{\mathbb{B}}
\newcommand{\Bthree}{\mathbb{B}_3}

\newcommand{\Inf}{\textrm{Inf}}
\renewcommand{\E}{\hat{F}}
\renewcommand{\dE}{\tilde{F}}
\newcommand{\nfaA}{\hat{\cA}}
\newcommand{\dfaA}{\tilde{\cA}}
\newcommand{\fsmA}{\cB}
\newcommand{\fQ}{\bar{Q}}
\newcommand{\fq}{\bar{q}}
\newcommand{\fdelta}{\bar{\delta}}
\newcommand{\fL}{\bar{\lambda}}

\newcommand{\Sem}[2]{{[#1 \models #2]}}
\newcommand{\SemPredPP}[2]{{[#1 \models _{\cPhat} #2]}}
\newcommand{\SemTracePP}[2]{{[#1 \models ^{t} #2]}}

\renewcommand{\phi}{\varphi}
\newcommand{\nphi}{{\neg\phi}}

\newcommand{\weg}[1]{}

\usepackage{listings}
\usepackage{color}

\lstdefinelanguage{forLTL}
{morekeywords={
  let, in, true, false,
  always, historically, alwaysinpast,
  between, betweeninpast,
  eventually, once, eventuallyinpast,
  from, after, frominpast,
  holding, holdinginpast,
  never, neverinpast,
  next, previous, nextinpast,
  nextn, previousn, nextninpast,
  occurring, occurringinpast,
  releases, triggered, releasesinpast,
  until, since, untilinpast,
  upto, before, uptoinpast,
  required, req, optional, opt, weak,
  inclusive, incl, exclusive, excl,
  and, or, implies, equals, not,
  if, then, else,
  rejecton, accepton,
  assert, declare, define,
  allof, someof, noneof, exactlyoneof,
  enumerate, list, as, in,
  with, without,
  timed,
  alw, even  %
},
sensitive=true,
morecomment=[l]{--},
morestring=[b]"
}

\begin{document}
\title{A Note on Runtime Verification of Concurrent Systems%
\thanks{This is a preliminary version. The final version is published within the proceedings of PFQA 2025 as LNCS by Springer.}
}

\author{Martin Leucker\orcidlink{0000-0002-3696-9222}}

\institute{%
Institute for Software Engineering and Programming Languages\\ University of Lübeck\\Germany \\
\email{leucker@isp.uni-luebeck.de}\\
}
\maketitle              %
\begin{abstract}
To maximize the information gained from a single execution when verifying a concurrent system, one can derive all concurrency-aware equivalent executions and check them against linear specifications. 
This paper offers an alternative perspective on verification of concurrent systems by leveraging trace-based logics rather than sequence-based formalisms. Linear Temporal Logic over Mazurkiewicz Traces (LTrL) operates on partial-order representations of executions, meaning that once a single execution is specified, all equivalent interleavings are implicitly considered. 
This paper introduces a three valued version of LTrL, indicating whether the so-far observed execution of the concurrent system is one of correct, incorrect or inconclusive, together with a suitable monitor synthesis procedure. To this end, the paper recalls a construction of trace-consistent Büchi automata for LTrL formulas 
and explains how to employ it in well-understood monitor synthesis procedures. In this way, a monitor results that yields for any linearization of an observed trace the same verification verdict.

\end{abstract}

\section{Introduction}

Following \cite{DBLP:journals/jlp/LeuckerS09}, \emph{runtime verification} is a lightweight verification technique that checks whether a system execution complies with a formally specified correctness property. Given such a property, a \emph{monitor} is synthesized, typically in the form of an automaton, which observes the system’s execution either offline or online (in real-time). In online runtime verification, the monitor incrementally processes observed events as the system runs and determines whether the execution satisfies or violates the correctness property. 
Unlike exhaustive verification techniques such as model checking, runtime verification provides \emph{on-the-fly analysis}, making it particularly useful for detecting issues in complex, concurrent, or distributed systems that are difficult to analyse statically. 

\emph{Concurrent systems} exhibit behaviours that are influenced by their \emph{environment}, including factors such as a \emph{scheduler}, \emph{system timing}, and \emph{resource availability}. In other words, unlike sequential systems, concurrent systems consist of \emph{independent threads or processes} that may share a processor, leading to executions where their actions occur in \emph{different orders} depending on scheduling decisions. This \emph{non-determinism} means that the same program can exhibit multiple behaviours across different executions. A fundamental property of concurrency is \emph{independence}: if two actions $a$ and $b$ are independent, observing execution $a$ followed by $b$ implies that an alternative execution order, $b$ followed by $a$, was also possible. This notion of \emph{equivalence between interleavings} is essential in reasoning about \emph{correctness}
in concurrent systems.

In a series of papers \cite{DBLP:journals/pacmpl/AngM24,DBLP:conf/cav/AngM24,DBLP:journals/pacmpl/0001P021}, runtime verification of concurrent systems was explored along the following lines. A correctness property\footnote{Correctness and incorrectness are dual notions here, so there is no significant difference in which one is explicitly defined.} $\phi$ is specified, which identifies when a single sequence is incorrect. Additionally, independence relations in concurrent systems were studied, allowing for the generalization from a single observation to potential alternative interleavings of the observed actions, as permitted by the given equivalence relation. 
The monitor reports success (or failure) if any of the considered interleavings are identified by the monitor.\footnote{\cite{DBLP:journals/pacmpl/AngM24,DBLP:conf/cav/AngM24,DBLP:journals/pacmpl/0001P021} use the term $\emph{predictive runtime verification}$ as further observations are \emph{predicted} from one observation. However, the term \emph{predictive runtime verification} was used with a different meaning in \cite{DBLP:conf/nfm/ZhangLD12}, where it referred to using the underlying program to predict how the currently observed behaviour might evolve. 
}

The general idea is sound; however, we propose to specify the correctness property while explicitly considering the concurrent nature of the system. Rather than defining correctness based on a particular execution order, the specification should be \emph{interleaving-independent}, ensuring that it captures the intended concurrent behaviour regardless of how independent actions are scheduled. The monitor, in turn, must account for all possible interleavings permitted by the system’s concurrency model. By doing so, the verification process remains robust, correctly identifying violations or confirmations of the correctness property across all valid execution orders. This paper provides an example of how to address the monitoring of concurrent systems using a trace logic that treats interleaving as a first-class citizen.  

We base our explanation on LTrL. Thiagarajan and Walukiewicz introduced LTrL, which is interpreted over partial-order representations of traces \cite{ThiagarajanW1997,DBLP:journals/iandc/ThiagarajanW02}. More specifically, it is defined over \emph{Mazurkiewicz traces} \cite{Mazurkiewicz88}, a special class of partial orders that respect independence between actions, as made precise in subsequent sections. It has been shown that LTrL is expressively equivalent to the first-order theory of traces when interpreted over (finite and) infinite traces.  

Each trace $T$ can be linearized in multiple ways, and we define $\text{lin}(T)$ as the set of all such linearizations. Clearly, $\text{lin}(T)$ is \emph{trace-closed}, meaning that it consists of all sequences obtained by taking one linearization of $T$ and applying all possible permutations of independent actions. Furthermore, LTrL provides a characterization of so-called \emph{trace-consistent} (or robust) LTL specifications. In other words, the models of any LTrL formula $\varphi$ define trace-closed languages.  

As such, our approach proceeds as follows:  
\begin{itemize}  
    \item Use LTrL to specify correctness properties to be monitored.  
    \item Synthesize a monitor that accepts all linearizations of traces satisfying the given correctness property.  
\end{itemize}  

While the general scheme from \cite{DBLP:conf/atva/DongLS08} is applicable to monitoring various linear-time temporal logics, particularly those that employ automata-based techniques to capture the models of a given formula, we explicitly provide the individual steps for clarity. Specifically, we build on existing methods to construct Büchi automata that accept all linearizations satisfying a given LTrL formula \cite{DBLP:conf/time/BolligL01,DBLP:journals/dke/BolligL03}.

Furthermore, we analyse the complexity of this approach. We show the complexity of monitors in non-elementary in the nesting of until-formulas, but, at the same time, optimal. Moreover, we discuss potential practical optimizations. 

The paper is organized as follows: We recall the concepts of words and automata in the next section. Mazurkiewicz traces are recalled in Section~\ref{sec:mazurkiewiczTraces} while LTrL for Mazurkiewicz traces is given in Section~\ref{sec:ltlForTraces}. The main contribution of the paper is provided in Section~\ref{sec:ltlThreeForTraces} which introduces a three-valued version of the LTrL together with a suitable monitor synthesis procedure. Section~\ref{sec:discussion} provides a short discussion on Mazurkiewicz trace logics in verification. In Section~\ref{sec:conclusion} we draw the conclusion of our approach and give directions for future work.

\section{Words and Automata}

For the remainder of this paper, 
let us fix an alphabet, i.e.\ a non-empty finite set, 
$\Sigma$. We write $a, a_i$ for any single
element of $\Sigma$ and sometimes call it an action.
Finite words over $\Sigma$ are elements of $\Sigma^\ast$, and are
usually denoted by $u, u', v, v', u_1, u_2, \dots$, whereas infinite words
are elements of $\Sigma^\omega$, usually denoted by $w, w', w_1, w_2,
\dots$.  We let $\Sigma^\infty$ denote the union of finite and infinite 
words. The empty word is denoted by $\epsilon$. 
Finally, we take $\prf(w)$ to be the set of finite
prefixes of $w$ and let $\letters{w}$ denote the set of actions
occurring in $w$.

A (nondeterministic) B\"uchi automaton (NBA) is a tuple $\cA =
(\Sigma, Q, Q_0, \delta, F)$, where $\Sigma$ is a finite alphabet, $Q$
is a finite non-empty set of states, $Q_0 \subseteq Q$ is a set of initial
states, $\delta: Q \times \Sigma \rightarrow 2^Q$ is the transition
function, and $F\subseteq Q$ is a set of accepting states. We extend
the transition function $\delta: Q \times \Sigma \rightarrow 2^Q$ to
sets of states and (input) words as
usual.
A \emph{run} of an automaton $\cA$ on a word $w = a_1\ldots \in
\Sigma^\omega$ is a sequence of states and actions $\rho =
q_0a_1q_1\dots$, where $q_0$ is an initial state of $\cA$ and for all
$i \in \N$ we have $q_{i+1} \in \delta(q_i, a_i)$.  For a run $\rho$,
let $\Inf(\rho)$ denote the states visited infinitely often. 
$\rho$ is called \Def{accepting} iff
$\Inf(\rho) \cap F \neq \emptyset$. 

A nondeterministic \Def{finite automaton} (NFA) $\cA = (\Sigma, Q,
Q_0, \delta, F)$, where $\Sigma$, $Q$, $Q_0$, $\delta$, and $F$
are defined as for a B\"uchi automaton, operates on finite
words. A \emph{run} of $\cA$ on a word $u = a_1\ldots a_n \in
\Sigma^\ast$ is a sequence of states and actions $\rho =
q_0a_1q_1\dots q_n$, where $q_0$ is an initial state of $\cA$ and for
all $i \in \N$ we have $q_{i+1} \in \delta(q_i, a_i)$.  The run is
called accepting if $q_n \in F$.
A NFA is called \Def{deterministic} and denoted DFA, iff for all $q
\in Q$, $a \in \Sigma$, $|\delta(q, a)| = 1$, and $|Q_0| = 1$.

As usual, the language accepted by an automaton (NBA/NFA/DFA), denoted
by $\cL(\cA)$, is given by its set of accepted words.

A \Def{Moore machine} (also
\Def{finite-state machine}, FSM) is a finite state automaton
enriched with output, formally denoted by a tuple $(\Sigma, Q, Q_0,
\delta, \Delta, \lambda)$, where $\Sigma$, $Q$, $Q_0 \subseteq Q$,
$\delta$ is as before and $\Delta$ is the output alphabet, $\lambda: Q
\rightarrow \Delta$ the output function.
As before, $\delta$ extends to the domain of words
as expected. 
Moreover, we denote by $\lambda$ also the
function that applied to a word $u$ yields the output in the state
reached by $u$ rather than the sequence of outputs.
 
\section{Mazurkiewic Traces}
\label{sec:mazurkiewiczTraces}

Following \cite{DBLP:journals/dke/BolligL03}, a \emph{(Mazurkiewicz) trace alphabet\/} is a pair $(\Sigma,I)$, where
$\Sigma$ is an alphabet and $I \subseteq \Sigma \times
\Sigma$ is an irreflexive and symmetric \emph{independence relation}.
Usually, $\Sigma$ consists of the \emph{actions} performed by a
distributed system while $I$ captures a static notion of causal
independence between actions. We define $D=(\Sigma \times \Sigma) - I$
to be the \emph{dependency relation}, which is then reflexive and
symmetric.

For the rest of the section, we fix a trace alphabet $(\Sigma,I)$. We
will use $a I b$ to denote that the actions $a$ and $b$ are
independent, \ie that $(a,b) \in I$, and use similar notation for
$(a,b)\in D$. We extend the notion to sets of actions $X, Y \subseteq
\Sigma$, and let $X I Y$ denote that each pair of actions
$a\in X$ and $b\in Y$ is independent. Moreover, $X D Y$ will denote
that $X$ is dependent on $Y$, \ie that there exists a pair of actions
$a \in X$ and $b \in Y$ with $a$ and $b$ dependent.  For convenience,
we will write $\{a\} I Y$ as $a I Y$ etc.

For the purpose of interpreting a linear temporal logice over traces, we will adopt the
viewpoint that traces are restricted labelled partial orders of events
and hence have an explicit representation of causality and
concurrency.

Let $T=(E,\leq,\lambda)$ be a $\Sigma$-labelled 
poset, \ie $(E,\leq)$ is a poset and $\lambda :
E \rightarrow \Sigma$ is a labelling function.
$\lambda$ can be extended to subsets of $E$ in the
straightforward manner. For $e\in E$, we define ${\downarrow} e=\{x\in
E \mid x\leq e \}$ and ${\uparrow} e=\{x\in E \mid e\leq x \}$.
We let $\lessdot$ be the \emph{covering
  relation} given by $x\lessdot y$ iff $x<y$ and for all $z\in E$,
$x\leq z\leq y$ implies $x=z$ or $z=y$.

A \emph{(Mazurkiewicz) trace} over $(\Sigma,I)$ is a $\Sigma$-labelled
poset $T=(E,\leq,\lambda)$ satisfying:
\begin{itemize}
\item
  $\View e$ is a finite set for each $e\in E$.
\item
  For every $e,e'\in E$, $e\lessdot e'$ implies $\lambda(e) \;D\;
  \lambda(e')$.
\item
  For every $e,e'\in E$, $\lambda(e) \;D\; \lambda(e')$ implies $e\leq
  e'$ or $e'\leq e$. 
\end{itemize} 

We let $\TR(\Sigma,I)$ denote the class of traces over
$(\Sigma,I)$. A trace language $L$ is a subset of traces,
\ie $L\subseteq\TR(\Sigma,I)$.  Throughout the paper, we will not
distinguish between isomorphic elements in $\TR(\Sigma,I)$. We will
refer to members of $E$ as \emph{events}.

Let $T=(E,\leq,\lambda)$ be a trace over $(\Sigma,I)$. 
A \emph{configuration} of $T$ is a finite
subset of events $c\subseteq E$ with ${\downarrow}c=c$ where
${\downarrow}c = \bigcup_{e \in c} {\downarrow}e$. The set of
configurations of $T$ will be denoted $\conf$. Trivially,
$\emptyset\in\conf$ for any trace 
$T$. %
$\conf$ can be equipped with a transition relation
${\longrightarrow_T}\subseteq\conf\times\Sigma\times\conf$ given by
$c\moveto{a} c'$ iff there exists an $e\in E$ such that
$\lambda(e)=a$, $e\not\in c$, and $c'=c\cup \{e\}$. Configurations of
$\conf$ are the trace-theoretic analogues of finite prefixes of
words. As will become apparent in Section~\ref{sec_logic}, the
formulas of our logic are to be interpreted at configurations of traces.

In its original formulation~\cite{Maz77}, Mazurkiewicz introduced
traces as certain equivalence classes of words, and this
correspondence turns out to be essential for our developments here. 
Let $T=(E,\leq,\lambda)\in \TR(\Sigma,I)$.  Then $w\in
\Sigma^\infty$ is a \emph{linearisation\/} of $T$ iff there exists a
map $\rho: \prf(w)\rightarrow \conf$ such that the following
conditions are met:
\begin{itemize}
\item
  $\rho(\varepsilon)=\emptyset$.
\item
  $\rho(v) \moveto{a} \rho(va)$ for each $va\in \prf(w)$.
\item
  For every $e\in E$, there exists some $u\in \prf(w)$ such that
  $e\in \rho(u)$. 
\end{itemize}
The function $\rho$ will be called a \emph{run map\/} of the
linearisation $w$. Note that the run map of a linearisation is unique.
In what follows, we shall take $\lin(T)$ to be the \Def{set of
  linearisations} of the trace $T$.

A set $p \subseteq \Sigma$ is called a \emph{$D$-clique\/} iff
$p \times p \subseteq D$. The equivalence relation ${\approx}
\subseteq \Sigma^{\infty} \times \Sigma^{\infty}$ induced by $I$ is
given by: $w \approx w'$ iff $w {\proj} p = w' {\proj} p$
for every $D$-clique $p$. Here and elsewhere, if $X\subseteq\Sigma$,
$w {\proj} X$ is the sequence obtained by erasing from $\w$ all
occurrences of letters in $\Sigma-X$. We take $[w]_{\approx}$ to
denote the ${\approx}$-equivalence class of $w\in\Sigma^\infty$.

It is not hard to show that elements of $\TR(\Sigma,I)$ and
${\approx}$-equivalence classes are two representations of the same
object: A labelled partial-order $T\in\TR(\Sigma,I)$ is represented by
$\lin(T)$ and vice versa (see~also\cite{DiekertRozenberg95}).  We exploit this
duality of representation and let $T_w$ denote the (unique) trace corresponding
to $[w]_{\approx}$. Moreover, for each $v\in\prf(w)$ we will use $c_v$
to denote the configuration of $\mathcal{C}_{T_w}$ given by $\rho(v)$.

\begin{figure}
  \begin{center}

\begin{tikzpicture}[scale=0.3, node distance=1mm, every node/.style={draw, minimum size=1.8em, inner sep=1pt}]

    \node (A0) at (2,2) {$a$};
    \node (B0) at (8,2) {$b$};
    \node (D0) at (5,6) {$d$};
    \node (A1) at (2,10) {$a$};
    \node (B1) at (8,10) {$b$};
    \node (B2) at (8,14) {$b$};
    \node (D1) at (5,18) {$d$};

    \draw[->] (A0) -- (D0);
    \draw[->] (B0) -- (D0);
    \draw[->] (D0) -- (A1);
    \draw[->] (D0) -- (B1);
    \draw[->] (A1) -- (D1);
    \draw[->] (B1) -- (B2);
    \draw[->] (B2) -- (D1);

    \draw[thick] plot [smooth] coordinates {(0,10.9) (5,12.2)(10,10.9)};
    \node[draw opacity=0] at (5,12.9) {$c$};

    \draw (0,18) -- (0,0) -- (10,0) -- (10,18);
    \draw[thick] (0,10.9) -- (0,0) -- (10,0) -- (10,10.9);
\end{tikzpicture}

  \caption{A trace over $(\Sigma,I)$.}
  \label{fig:trace}
 
  \end{center}
\end{figure}
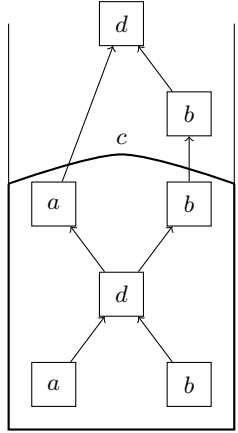

To illustrate these concepts, consider the trace alphabet $(\Sigma,I)$
with $\Sigma=\{a,b,d\}$ and $I=\{(a,b),(b,a)\}$. An example trace $T$
over $(\Sigma,I)$ is depicted in Figure~\ref{fig:trace} with smaller
elements (with respect to ${\leq}$) appearing below larger elements.
Furthermore, it can easily be verified that $abdbabd\in\lin(T)$ so
$T=T_{abdbabd}$, but $adabbbd\not\in\lin(T)$. The configuration
$c\in\conf$ consists of the first two $a$'s, first $d$, first two
$b$'s, and is also denoted by $c_{abdab}$, which is identical to
$c_{badab}$ as $abdab\approx badab$.

We transfer considering traces as equivalence classes to the level of
languages and call a word language $L \subseteq \Sigma^\omega$
\Def{trace closed} iff for all words $w,w' \in \Sigma^\omega$ with
$w \approx w'$, it holds $w \in L$ iff $w' \in L$.

Finally, we call a B\"uchi automaton, NFA, DFA etc.\ $\cA$
\Def{trace closed} if its language $\mathcal{L}(\cA)$ is
trace closed. 

\section{LTL for Mazurkiewicz Traces}

\label{sec_logic}
\label{sec:ltlForTraces}

In this section, we bring out the syntax and semantics of the linear
temporal logic LTL, which will be our basic object of study.  It was
originally introduced for words by Pnueli~\cite{Pnueli77}. It was
later equipped with a trace semantics~\cite{ThiagarajanW1997} and some additional operators, and termed LTrL.  Diekert and Gastin~\cite{DiekertGastin00} have shown that LTL with the same trace semantics but without these additional operators is already expressively equivalent to first-order logic for traces. As such, we will use this version in the following. Note that we use LTrL to highlight that we consider Mazurkiewicz traces but LTL when considering words. However, formally one can combine the two approaches using a parameterized version of LTL:

The formulas of $\LTL$ are parameterised by a trace alphabet
$(\Sigma,I)$ and are defined inductively as follows:
$$
\LTL(\Sigma,I) :: = \True \mid \nega\phi \mid \phi \vee \psi\mid
\nexta{a}\phi \mid \phi \Uu \psi~,~~a\in\Sigma.
$$
\noindent Formulas of $\LTL(\Sigma,I)$ are interpreted over
configurations of traces over $(\Sigma,I)$. More precisely,
given a trace $T\in\TR(\Sigma,I)$, a configuration
$c\in\mathcal{C}_T$, and a formula $\phi\in\LTL(\Sigma, I)$,
the notion of $T, c \models \phi$ is defined inductively via:
\begin{itemize}
\item $T,c \models \True$.
\item $T,c \models \nega\phi$ iff $T,c \not\models \phi$.
\item $T,c \models \phi \vee \psi$ iff $T,c \models \phi$ or $T,c
  \models \psi$.
\item $T,c \models \nexta{a}\phi$ iff there exists a $c'\in\conf$
  such that $c\moveto{a}c'$ and $T,c' \models \phi$.
\item $T,c \models \phi\Uu\psi$ iff there exists a $c'\in\conf$
  with $c\subseteq c'$ such that $T, c' \models \psi$ and all $c''
  \in \conf$ with $c \subseteq c'' \subset c'$ satisfy $\phi$.
\end{itemize}
\noindent
We will freely use the standard abbreviations such as e.g.
$\False=\neg\True$, $\varphi\wedge\psi = \neg(\neg\varphi \vee
\neg\psi)$. Furthermore, we sometimes abbreviate $T, \emptyset \models
\phi$ by $T \models \phi$. All models of a formula $\phi \in
\LTL(\alphabet,\ir)$ constitute a subset of $\TR(\Sigma,I)$, thus a
language. It is denoted by $\langP{\phi}$ and is called the language
\Def{defined} by $\phi$. Furthermore, every formula defines an
$\omega$-language \viz the set $\bigcup \{ \linP{\T} \mid \T \models
\phi \}$, which is also indicated by $\langP{\phi}$.

A simple example of a formula of LTL is
$\varphi=\nexta{a}\nexta{b}\psi$. Note that for the trace from
Figure~\ref{fig:trace}, it holds that $T \models \varphi$ if and only
if $T,\con{ab} \models \psi$. Moreover, $\varphi$ is equivalent to
$\varphi'=\nexta{b}\nexta{a}\psi$ over this particular trace alphabet
because $a I b$, \ie the models of $\varphi$ and $\varphi'$ 
coincide.

We note that, in case of the empty independence relation,
$\LTL(\Sigma, I)$ and $\LTL$ interpreted over words 
coincide in the expected manner. Thus, we identify LTL
over words with $\LTL(\Sigma, \emptyset)$
and save the work of formally introducing LTL over
words.\footnote{Typically, LTL is introduced using atomic propositions rather than labelled actions. This difference should not bother us any further.}

Note that an LTL formula $\phi$ (considered over the alphabet $(\Sigma, \emptyset)$) is called \emph{trace consistent wrt.\ the trace alphabet $(\Sigma,I)$}, if $\cL(\phi)$ is trace closed wrt.\ $(\Sigma, I)$. In other words, $\phi$ is called trace consistent iff it cannot distinguish between two equivalent linearizations of a trace. All LTL$(\Sigma, I)$ formulas are trace consistent by definition. In simple words, every  LTL$(\Sigma, I)$ formula respects the given concurrency relation.

We conclude this section recalling the result that each $\LTL(\Sigma, I)$ formula $\phi$ can be translated into a corresponding (alternating) Büchi automaton accepting all linearizations of all traces satisfying $\phi$. As such, the language is especially trace closed.\footnote{While a non-deterministic automaton maps a given state and an input action to a set of possible next states, an alternating automaton maps to a Boolean combination of states \cite{Chandra:1981:A}. It is easy to see that every alternating Büchi automaton can be translated into a non-deterministic Büchi automaton accepting the same language, involving an exponential blow-up \cite{MiyanoH84}. While in the seminal work by Vardi and Wolper \cite{VardiWolper86}, an LTL formula is translated into a non-deterministic Büchi automaton directly, Vardi presented a translation via alternating automata \cite{Vardi96}, which, arguably, conceptually simplifies the translation into two steps: Translating LTL into an automaton that supports conjunctions of states, and eliminating conjunction when translating alternation into non-determinism. \cite{DBLP:journals/dke/BolligL03} follows the latter route, but adapted to LTrL. Hence, strictly speaking \cite{DBLP:journals/dke/BolligL03} only states the translation result into alternating machines, which can subsequently be translated to non-determinis Büchi automata.}

\begin{theorem}[\cite{DBLP:journals/dke/BolligL03}]
  \label{thm_main}
  Let $\phi$ be a formula of $\LTL(\Sigma, I)$. There is an (alternating)  B\"uchi automaton $\cA_\phi$ such that  $\langP{\cA_\phi} =
  \langP{\phi}$. The size of $\cA_\phi$ is non-elementary in the size of $\phi$. 
\end{theorem}

\section{Runtime verification for LTrL}
\label{sec:ltlThreeForTraces}

Let us now introduce runtime verification for LTrL following the anticipatory approach originally introduced in \cite{BauerLS06,DBLP:journals/tosem/BauerLS11}. To this end, let us recall the 3-valued approach for LTL over words first.

\paragraph{Monitoring \LTLthree over words.}

Let us recall our 3-valued semantics, denoted by  \LTLthree, over the
set of truth values $\Bthree = \{ \bot, ?, \top\}$ from \cite{BauerLS06,DBLP:journals/tosem/BauerLS11}:
  Let $u \in \Sigma^\ast$ denote a finite word.  The \emph{truth value}
  of a \LTLthree formula $\varphi$ \wrt $u$, denoted by $\Sem{u}{\phi}$,
  is an element of $\Bthree$ defined by%
  \begin{equation*}
    \Sem{u}{\phi} = 
    \left\{
      \begin{array}{l@{\quad}l}
        \top    & \textrm{if}\ \forall{\sigma\in\Sigma^\omega}: u\sigma \models \varphi\\[0.5ex]
        \bot    & \textrm{if}\ \forall{\sigma\in\Sigma^\omega}: u\sigma \not\models \varphi\\[0.5ex]
        ?       & \textrm{otherwise}.
      \end{array}
    \right.
  \end{equation*} %

Monitor synthesis was introduced for \LTLthree in \cite{BauerLS06,DBLP:journals/tosem/BauerLS11}, and the ideas were later generalized in a systematic way to classes of linear temporal logics satisfying several additional properties in \cite{DBLP:conf/atva/DongLS08-double}. Since LTrL meets all these criteria, we can obtain a monitor for LTrL using this approach. However, to simplify the presentation, we directly illustrate how the construction for LTrL aligns with that of \LTLthree, rather than recalling the more general framework from \cite{DBLP:conf/atva/DongLS08-double}.  

As such, let us first recall the monitor synthesis approach along the lines of \cite{BauerLS06,DBLP:journals/tosem/BauerLS11}: For a given formula $\phi \in \LTL$, we construct a finite Moore
machine (FSM), $\fsmA^\phi_\cPhat$ that reads finite words $u \in
\Sigma^\ast$ and outputs $\Sem{u}{\phi} \in \Bthree$.

For an NBA $\cA$, we denote by $\cA(q)$ the NBA that coincides with
$\cA$ except for $Q_0$, which is defined as $Q_0 =
\{q\}$.
Fix $\phi \in \LTL$ for the rest of this paragraph and let $\cA^\phi=(\Sigma, Q^\phi, Q_0^\phi, \delta^\phi, F^\phi)$
denote the NBA, which accepts all models of $\phi$, and let
$\cA^\nphi=(\Sigma, Q^\nphi, Q_0^\nphi, \delta^\nphi, F^\nphi)$ denote the NBA, which accepts all counter examples of
$\phi$. The corresponding construction is standard
\cite{VardiWolper86}.

  For $u \in \Sigma^\ast$ and $\delta(Q_0^\phi, u) = \{ q_1, \dots,
  q_l \}$, we have 
 $ \Sem{u}{\phi} \neq \bot \mbox{ iff } \exists q \in \{q_1, \dots,
  q_l\}\ \text{ such that } \mathcal{L}(\cA^\phi(q)) \neq \emptyset.
$
Likewise, for the NBA $\cA^\nphi$, we have 
  for $u \in \Sigma^\ast$, and $\delta(Q_0^\nphi, u) = \{ q_1, \dots,
  q_l \}$ that
  $\Sem{u}{\phi} \neq \top \mbox{ iff } \exists q \in \{q_1, \dots,
  q_l\}$ such that $\mathcal{L}(\cA^\nphi(q)) \neq \emptyset.$

Following \cite{DBLP:journals/tosem/BauerLS11}, for $\cA^\phi$ and $\cA^\nphi$, we
now define a function $\cF^\phi : Q^\phi \to \B$ respectively
$\cF^\nphi : Q^\nphi \to \B$ (where $\B = \{\top, \bot\}$), assigning
to each state $q$ whether the language of the respective automaton
starting in state $q$ is not empty.
Thus, if $\cF^\phi(q)=\top$ holds, then the automaton $\cA^\phi$
starting at state $q$ accepts a non-empty language and each finite
prefix $u$ leading to state $q$ can be extended to an (infinite) run
satisfying $\phi$.

Using $\cF^\phi$ and $\cF^\nphi$, we turn from automata over infinite words to automata over finite words: We define two NFAs $\nfaA^\phi =
(\Sigma, Q^\phi, Q_0^\phi, \delta^\phi, \E^\phi)$ and $\nfaA^\nphi =
(\Sigma, Q^\nphi, Q_0^\nphi, \delta^\nphi, \E^\nphi)$ where
$\E^\phi = \{ q \in Q^\phi \mid \cF^\phi(q) = \top\}$ and 
$\E^\nphi = \{ q \in Q^\nphi \mid \cF^\nphi(q) = \top\}$.
Then, 
we have for all $u\in \Sigma^\ast$:%
\[
 u \in \mathcal{L}(\nfaA^\phi) \mbox{ iff }\Sem{u}{\phi} \neq \bot
 \quad\mbox{and}\quad
 u \in \mathcal{L}(\nfaA^\nphi) \mbox{ iff }\Sem{u}{\phi} \neq \top
\]
Hence, we can evaluate $\SemPredPP{u}{\phi}$ 
as follows:
\label{lem:eval}
We have 
$\Sem{u}{\phi} = \top$ if $u \not\in \mathcal{L}(\nfaA^\nphi)$,
$\Sem{u}{\phi} = \bot$ if $u \not\in \mathcal{L}(\nfaA^\phi)$, and
$\Sem{u}{\phi} = ?$ if $u \in \mathcal{L}(\nfaA^\phi) \mbox{ and } u
\in \mathcal{L}(\nfaA^\nphi)$.

As a final step, we now define a (deterministic) FSM $\fsmA^\phi$ that
outputs for each finite string $u$ and formula $\phi$ its three valued semantics.
Let $\dfaA^\phi$ and $\dfaA^\nphi$ be the deterministic versions of
$\nfaA^\phi$ and $\nfaA^\nphi$, which can be computed in
the standard manner by power-set construction.
Now, we define the FSM in question as a product of $\dfaA^\phi$ and
$\dfaA^\nphi$:

\begin{definition}[Monitor $\fsmA^\phi$ for \LTL-formula $\phi$]
  \label{def:monitor}
  Let $\dfaA^\phi=(\Sigma, Q^\phi, \{q_0^\phi\},$ $\delta^\phi, \dE^\phi)$
  and $\dfaA^\nphi=(\Sigma, Q^\nphi, \{q_0^\nphi\}, \delta^\nphi,
  \dE^\nphi)$ be the DFAs which correspond to the two NFAs $\nfaA^\phi$
  and $\nfaA^\nphi$ as defined before. %
  Then we define the \emph{monitor} $\fsmA^\phi=\dfaA^\phi
  \times \dfaA^\nphi$ for $\phi$ as the
  minimized version of the FSM $(\Sigma, \fQ, \fq_0, \fdelta, \fL)$,
  where $\fQ = Q^\phi \times Q^\nphi$, $\fq_0 = (q_0^\phi,
  q_0^\nphi)$, $\fdelta((q,q'), a) = (\delta^\phi(q, a),
  \delta^\nphi(q',a))$, and $\fL : \fQ \to \Bthree$ is defined
  by%
    \[ 
    \fL((q,q')) = \left\{
      \begin{array}{l@{\quad}l}
        \top & \mbox{if } q' \not\in \dE^\nphi\\
        \bot & \mbox{if } q  \not\in \dE^\phi\\
        ?    & \mbox{if } q \in \dE^\phi \mbox{ and } q' \in \dE^\nphi.
      \end{array}
    \right.
    \]
\end{definition}
We sum up our entire construction in Fig.~\ref{tab:proc} and conclude
with the following correctness theorem.
\begin{theorem}[\cite{BauerLS06,DBLP:journals/tosem/BauerLS11}]
  \label{thm:untimed}
  Let 
  $\phi \in \LTL$, and let $\fsmA^\phi = (\Sigma, \fQ, \fq_0, \fdelta,
  \fL)$  be the corresponding monitor.
  Then, for all $u \in \Sigma^\ast$: %
  $\Sem{u}{\phi} = \fL(\fdelta(\fq_0,u))$.
\end{theorem}

\tikzstyle{decision} = [diamond, draw, fill=blue!20, 
    text width=4.5em, text badly centered, node distance=3cm, inner sep=0pt]
\tikzstyle{block} = [rectangle, draw, fill=blue!10, 
    text width=5em, text centered, rounded corners, minimum height=3em]
\tikzstyle{blob} = [circle, draw, fill=blue!10, 
    text width=1.8em, text centered, rounded corners, minimum height=1.8em]
\tikzstyle{line} = [draw, -latex']
\tikzstyle{cloud} = [draw, ellipse,fill=red!20, node distance=3cm,
    minimum height=2em]
    
\newcommand{\algUntimed}{%
\scalebox{0.7}{%
\begin{tikzpicture}[node distance = 2cm, auto]
    \node [blob] (phiInit) {$\phi$};
    \node [blob, yshift = 5.5mm, right of=phiInit] (phi) {$\phi$} 
       edge [<-] (phiInit);
    \node [blob, right of=phi] (Aphi) {$\cA_\phi$}  edge [<-] (phi);
    \node [blob, right of=Aphi] (Fphi) {$\cF_\phi$}  edge [<-] (Aphi);
    \node [blob, right of=Fphi] (AHphi) {$\hat\cA_\phi$}  edge [<-] (Fphi);
    \node [blob, right of=AHphi] (ATphi) {$\tilde\cA_\phi$}  edge [<-] (AHphi);
    \node [blob, yshift = -5.5mm, right of=phiInit] (nphi) {$\nphi$}  edge [<-] (phiInit);
    \node [blob, right of=nphi] (Anphi) {$\cA_\nphi$}  edge [<-] (nphi);
    \node [blob, right of=Anphi] (Fnphi) {$\cF_\nphi$}  edge [<-] (Anphi);
    \node [blob, right of=Fnphi] (AHnphi) {$\hat\cA_\nphi$}  edge [<-] (Fnphi);
    \node [blob, right of=AHnphi] (ATnphi) {$\tilde\cA_\nphi$}  edge [<-] (AHnphi);
    \node [blob, yshift = -5mm, right of=ATphi] (Mphi) {$\cM_\phi$} 
    edge [<-] (ATphi)
    edge [<-] (ATnphi);
    \node [block,  yshift = 5.5mm, above of=phiInit] (Input) {Input};
    \node [block, above of=phi] (Formula) {Formula};
    \node [block, above of=Aphi] (NBA) {NBA};
    \node [block, above of=Fphi] (EPS) {Emptiness per state};
    \node [block, above of=AHphi] (NFA) {NFA};
    \node [block, above of=ATphi] (DFA) {DFA};
    \node [block,  yshift = 5.5mm, above of=Mphi] (FSM) {FSM};
\end{tikzpicture}
}}

 \begin{figure}[t]
\begin{center}
  \algUntimed\vspace{-3.5ex}
\end{center}
\caption{The procedure for getting $\Sem{u}{\phi}$ for a given
  $\phi$.\label{tab:proc}}%
\end{figure}

\paragraph{Monitoring \LTLthreeTraces over Mazurkiewicz Traces.}

Let us now define our 3-valued semantics, denoted by  \LTLthreeTraces, with the
set of truth values $\Bthree = \{ \bot, ?, \top\}$ similar as above: 
\begin{definition}[3-valued semantics of LTrL]
  \label{def:ltl3semantics}
  Let $u \in \Sigma^\ast$ denote a finite word.  The \emph{truth value}
  of a \LTLthreeTraces formula $\varphi$ \wrt $u$, denoted by $\SemTracePP{u}{\phi}$,
  is an element of $\Bthree$ defined by%
  \begin{equation*}
    \SemTracePP{u}{\phi} = 
    \left\{
      \begin{array}{l@{\quad}l}
        \top    & \textrm{if}\ \forall{\sigma\in\Sigma^\omega}: T_{u\sigma} \models \varphi\\[0.5ex]
        \bot    & \textrm{if}\ \forall{\sigma\in\Sigma^\omega}: T_{u\sigma} \not\models \varphi\\[0.5ex]
        ?       & \textrm{otherwise}.
      \end{array}
    \right.
  \end{equation*} %
\end{definition}

\begin{lemma}
 3-valued LTrL coincides with 3-valued LTL when all actions are independent.
\end{lemma}

Assume we have observed a sequence of actions $u$. This yields the configuration $c_u$ of any trace $T$ with $c_u \subset T$. As $c_u = c_{v}$ for any $u \approx v$, we get that, for any $w\in\Sigma^\omega$, that $T_{uw} = T_{vw}$ and thus $T_{uw} \models \phi$ iff $T_{vw} \models \phi$.
In other words, we get that \LTLthreeTraces gives the same verdict of all linearizations of $c_u$ and thus for all words equivalent to $u$.

\begin{theorem}
Let $\phi \in \LTLthreeTraces$ and $u \in \Sigma^\ast$, 
  \[ \SemTracePP{u}{\phi} = \SemTracePP{v}{\phi} \]
  for all $v \in [u]_{\approx}$. 
\end{theorem}

In other words, any monitor following the semantics of \LTLthreeTraces 
yields the same value for $v$ if seen $u$, when $v \approx u$.

Now, we turn our focus to the monitoring procedure. Our main goal is to understand that we can follow the approach for \LTLthree but starting with the Büchi automaton for an \LTLthreeTraces formula according to Theorem~\ref{thm_main}. To this extent, let us consider the automaton for any LTrL formula in more detail.

Reading a finite linearization $u$ yields sets of states in the automaton. The languages accepted from that states characterize together with the prefix $u$ the traces satisfying $\phi$ with having $c_u$ as configuration. 

We call this language $L_u$.

\begin{lemma}
$L_u = L_v$, if $u$ and $v$ are equivalent. 
\end{lemma}

\begin{proof}
Assume the contrary: Then there is $w$ such that, wlog, $w$ is in $L_u$ but not in $L_v$. As such $uw \in \cL(\cA)$ but $vw \notin \cL(\cA)$. But, as $u \approx v$, we have $T_{uw} \models \phi$ iff $T_{vw} \models \phi$. As $\cA$ is trace closed, $uw \in \cL(\cA)$ iff $vw \in \cL(\cA)$. Contradiction.
\end{proof}

Note that the idea of solely relying on the states reached when reading a prefix $u$, was called \emph{forgettable past} in \cite{DBLP:conf/atva/DongLS08-double}.

\begin{corollary}
$L_u = \emptyset$ iff for all $w$ we have $T_{uw} \not \models \phi$.
\end{corollary}

In other words, we can follow the schema for \LTLthree to build a monitor for LTrL. We follow the construction shown in Fig.~\ref{tab:proc} and conclude, using the notation before, 
with the following correctness theorem.
\begin{theorem}
  \label{thm:traces}
  Let 
  $\phi \in \LTLthreeTraces$, and let $\fsmA^\phi = (\Sigma, \fQ, \fq_0, \fdelta,
  \fL)$  be the corresponding monitor.
  Then, for all $u \in \Sigma^\ast$: %
  $\SemTracePP{u}{\phi} = \fL(\fdelta(\fq_0,u))$.
\end{theorem}

Since the Büchi automaton for an LTL formula over traces is necessarily non-elementary \cite{Walukiewicz98}, the resulting monitor is also non-elementary. When the independence alphabet is empty, the construction from \cite{DBLP:journals/dke/BolligL03} aligns with traditional LTL methods, producing a Büchi automaton of single exponential size and, consequently, a monitor of double exponential size. Regardless of the case, since the final FSM is minimized, all monitors obtained through this approach are optimal in terms of size.

\section{Discussion}
\label{sec:discussion}

Trace-consistent properties have been successfully applied in model checking to enable \emph{partial order reduction}. The core idea of partial order reduction in model checking is to consider only a single execution (or a small representative subset) for all executions that are equivalent up to permutation of independent actions. Simply put, when actions $a$ and $b$ are independent, one considers only executions where $a$ happens before $b$, rather than also including the reverse order. Fixing such an order not only reduces the number of executions that need to be explored in the underlying system but also allows for pruning the automaton that accepts models of the underlying correctness specification.  

In runtime verification, where the underlying system is not under control, the monitor must be prepared to handle any possible ordering of independent actions. In other words, model checking is a \emph{white-box} technique that allows for more optimizations, whereas runtime verification operates as a \emph{black-box} technique, limiting the available reductions.  Still, it might be worthwhile to consider potential optimizations following ideas of reordering actions.

An alternative approach to the one presented here\footnote{This remark was communicated by one of the anonymous reviewers of this paper.} is to use an algorithm that checks whether a specification is Mazurkiewcz-trace closed. Then, one can simply use LTL and traditional automata constructions. Note that, as the final monitor is a minimized, unique Moore machine, the originating LTL formula may be non-elementary longer than an equivalent LTrL formula. In practice, readability may be of great importance and it is unclear whether compact formulas may be easier or more difficult to understand. A notable reference for checking whether an LTL specification is trace closed and also considers more general equivalence relations is \cite{PeledWilWol98}.

The method proposed here may work similarly for different notions of concurrency and corresponding notions of equivalences, suitable logics, and automata-based decision procedures, e.g.\ \cite{AlurPP95,Alur:1998:DGP,Thiagarajan94}.

\section{Conclusion}
\label{sec:conclusion}

In this paper, we have explored the use of (Mazurkiewicz) trace logics for monitoring concurrent systems. By incorporating concurrency as a first-class citizen within temporal logic, we have demonstrated that it is possible to obtain the same verdict based on any linearization that could have occurred instead of the observed one, thanks to the independence of actions. 

Several directions remain for future research. A key next step is the practical implementation and evaluation of our approach in real-world systems to assess its efficiency and applicability. Additionally, a more detailed comparison with existing methods that rely on generalizations of executions rather than built-in concurrency will help to clarify the advantages of our approach. Beyond this, we aim to explore past-time logics for Mazurkiewicz traces, as well as logics that extend beyond Mazurkiewicz traces to model concurrency in a broader range of systems.

\bibliographystyle{splncs04}
\bibliography{mybibliography}
\end{document}